\documentclass[11pt,letterpaper]{article}
\usepackage[utf8]{inputenc}

\pdfoutput=1

\usepackage[margin=1in]{geometry}
\usepackage{amsmath}
\usepackage{palatino}
\usepackage{macros}
\usepackage{latexsym}

\usepackage{tikz}
\usetikzlibrary{arrows,positioning, shapes.geometric}

\usepackage{enumitem}

\usepackage{thmtools} 
\usepackage{thm-restate}

\allowdisplaybreaks
        
\bibliography{refs}

\newif\ifauthors
\authorstrue

\newcommand{\dimension}{d}

\newcommand{\declareperson}[1]{\expandafter\newcommand\csname#1\endcsname[1]{\textcolor{orange}{#1: ##1}}}
\DeclareMathOperator{\NBC}{BC}

\ifauthors
\declareperson{Shayan}
\declareperson{Dorna}
\declareperson{Kasper}
\def\ord{{\mathcal{O}}}
\usepackage{authblk}

 \author{Dorna Abdolazimi\thanks{\href{mailto:dornaa@cs.washington.edu}{dornaa@cs.washington.edu}. Research supported by NSF grant  CCF-2203541, and Air Force Office of Scientific Research grant FA9550-20-1-0212.}}
 \author{Kasper Lindberg \thanks{\href{mailto:whet@cs.washington.edu}{whet@cs.washington.edu}. Research supported by NSF grant CCF-2203541.}} 
\author{Shayan Oveis Gharan\thanks{\href{mailto:shayan@cs.washington.edu}{shayan@cs.washington.edu}. Research supported by NSF grant CCF-2203541, Air Force Office of Scientific Research grant FA9550-20-1-0212, and Simons Investigator grant.}}%
\affil{University of Washington}
\fi

\title{On Optimization and Counting of\\ Non-Broken Bases of Matroids}

\begin{document}
\maketitle
\begin{abstract}
    Given a matroid $M=(E,{\cal I})$, and a total ordering over the elements $E$, a broken circuit is a circuit where the smallest element is removed and an NBC independent set is an independent set in ${\cal I}$ with no broken circuit. The set of NBC independent sets of any matroid $M$ define a simplicial complex called the broken circuit complex which has been the subject of intense study in combinatorics. Recently, Adiprasito, Huh and Katz showed that the face of numbers of any broken circuit complex form a log-concave sequence, proving a long-standing conjecture of Rota.

    We study counting and optimization problems on NBC bases of a generic matroid. We find several fundamental differences with the independent set complex: for example, we show that it is NP-hard to find the max-weight NBC base of a matroid or that the convex hull of NBC bases of a matroid has edges of arbitrary large length.
    We also give evidence that the natural down-up walk on the space of NBC bases of a matroid may not mix rapidly by showing that for some family of matroids it is NP-hard to count the number of NBC bases after certain conditionings. 
\end{abstract}    
\section{Introduction}
A matroid $M = (E, \mathcal{I})$ is consists of a finite ground set $E$ and a collection $\mathcal{I}$ of subsets of $E$, called independent sets, satisfying:
\begin{description}
    \item [Downward closure:] If $S \subseteq T$ and $T \in \mathcal{I}$, then $S \in \mathcal{I}$.
    \item [Exchange axiom:] If $S, T \in \mathcal{I}$ and $|T| > |S|$, then there exists an element $i \in T \setminus S$ such that $S \cup \{i\} \in \mathcal{I}$.
\end{description}
The rank of a set $S \subseteq E$ is the size of the largest independent set contained in $S$. All maximal independent sets of $M$, called the bases of $M$, have the same size $r$, which is called the rank of $M$. 

Sampling and counting problems on matroids have captured the interest of many researchers for several decades with applications to reliability \cite{CP89}, liquidity of markets \cite{Ramseyer_2020}, etc. A recent breakthrough in this field proved that the down-up walk on the bases of a matroid mixes rapidly to the (uniform) stationary distribution and can be used to count the number of bases of a matroid \cite{ALOV19, cryan2020modified}, resolving the conjecture of Mihail and Vazirani from 1989 \cite{MV89}. 
The down-up walk is easy to describe: Start with an arbitrary base $B$ and repeatedly execute the following two steps:
\begin{enumerate}
\item Choose a uniformly random element $i\in B$ and delete it.
\item Among all bases (of $M$) that contain $B\smallsetminus \{i\}$, choose one uniformly at random.
\end{enumerate}

A central question that has puzzled researchers since then is sampling a non-broken (circuit) basis (NBC basis) of a matroid \cite{benson2010gparking}.
A set  $C \subseteq E$ is a {\it circuit} iff  $C \setminus\{e\} \in \mathcal{I}$ for any $e \in C$. 
A {\it broken circuit} (with respect to a total ordering $\ord$) is a set $C \setminus \{e\}$, where $C \subseteq E$ is a circuit and $e$ is the {\bf smallest} element of $C$ with respect $\ord$. An independent set $S \subseteq E$ is a {\it non-broken}  independent set (NBC independent set) if it contains no broken circuits. The  NBC independent sets are closely related to several interesting combinatorial objects. The number of NBC independent sets of size $k$ in a graphic matroid is equal to the absolute value of the $(n-1) -k$-th coefficient of the chromatic polynomial of the underlying graph where $n$ is the number of vertices. As a corollary the following facts hold:

\begin{fact}\label{fact:NBCcons}The following facts are well-known about the counts of NBC bases/independent sets of different family of matroids. 
\begin{itemize}
\item The number of all NBC independent sets of a graphic matroid is equal to the the  number of acyclic orientations of the graph \cite{STANLEY1973171}.
\item The number of all NBC independent sets of a co-graphic matroid is equal to the number of strongly connected orientations of the graph (see e.g., \cite{GIOAN2019165}).
\item The number of non-broken spanning trees of a graph is equal to the number of parking functions with respect to a unique source vertex \cite{benson2010gparking}
\item The number of NBC independent of sets of linear matroid with vectors $v_1,\dots,v_n$ is equal to the number of regions defined by the intersection of the orthogonal hyperplanes (see e.g., \cite{Stanley2007AnIT}).
\end{itemize}
\end{fact}

We emphasize that although the set of NBC independent sets/bases of a matroid are functions of the underlying total order $\ord$, the counts of the number NBC independent sets of rank $k$ for any $0\leq k\leq r$ are invariant under $\ord$ \cite{Stanley2007AnIT}. We remark that, to the best of our knowledge as of this date, none of the above counting problems are known to be computationally tractable.

Given a matroid $M$ with an arbitrary total ordering $\ord$, one can analogously run the down-up walk only on the NBC bases of $M$. It is not hard to see that this chain is irreducible and converges to the uniform stationary distribution. Following the work of \cite{ALOV19} it was conjectured that the down-up walk on the NBC bases of any matroid mixes rapidly \footnote{In fact, this conjecture was raised an an open problem in several recent workshops \href{https://sites.cs.ucsb.edu/~vigoda/School/}{UC Santa Barbara workshop on New tools for Optimal Mixing of Markov Chains: Spectral Independence and Entropy Decay}, and \href{https://simons.berkeley.edu/programs/geometry-polynomials/}{Simon's workshop on Geometry of Polynomials}}.

\begin{conjecture}\label{conj:mixingdownupNBC}
For any matroid $M$, and any total ordering $\ord$ of the elements of $M$, the down-up walk on the NBC bases of a matroid mixes in polynomial time.
\end{conjecture}

It turns out that the above conjecture, if true, would be give a generalization of the result of \cite{ALOV19}, because of the following fact.

\begin{fact}[\cite{Brylawski77}]\label{fact:nbctomatroid}
    For any matroid $M$ one can construct another matroid $M'$ with an ordering $\ord$ with only one extra element such that there is a bijection between bases of $M$ and non-broken bases of $M'$.
\end{fact}

Furthermore, if the above conjecture is true, then since matroids are closed under truncation, one can also count the number of all NBC independent sets of $M$, thus resolving all of the open problems in \cref{fact:NBCcons}.

A promising reason to expect these problems to be tractable in the first place is the remarkable work of Adiprasito, Huh and Katz \cite{adiprasito2018hodge} who proved the Rota's conjecture showing that the face numbers of a broken circuit complex (see below for definition) of any matroid forms a log-concave sequence. For comparison, it is well-known that the coefficients of the matching polynomial of any graph form a log-concave sequence and the classical algorithm of Jerrum-Sinclair \cite{doi:10.1137/0218077} gives an efficient algorithm to count the number of matchings of any graph (although to this date we still don't know an efficient algorithm to count the number of {\bf perfect} matchings of general graphs).

\subsection{Background}
The existing analyses of the mixing time of the down-up walk for bases of matroids, crucially rely on the theory of high dimensional simplicial complxes  \cite{ALOV19, KO18}, which has found many intriguing applications in several areas of computer science and math in the past few years \cite{gotlib2023high}. 

{\bf Simplicial Complex.} A {\it simplicial complex} $X$ on a finite ground set $U$ is a downwards closed set system, i.e. if $\tau \in X$ and $\sigma \subset \tau \subseteq U$, then $\sigma \in X$. The elements of $X$ are called faces, and the maximal faces are called facets.  We say $X$ is a pure $d$-dimensional complex if all of its facets are of  size $d$. We denote the set of facets by $X(d)$. 
 A weighted simplicial complex $(X,\pi)$ is a simplicial complex $X$ paired with a probability distribution $\pi$ on its facets. The global walk (down-up walk) $P^{\vee}$ on the facets of a $d$-dimensional complex $(X,\pi)$  is defined as follows: starting at a facet $\tau$, we transition to the next facet $\tau'$ by the following two steps:
\begin{enumerate}
    \item Select a uniformly random element $x \in \tau$ and remove $x$ from $\tau$.
    \item Select a random facet $\tau'$ containing $\tau \setminus \{x\}$ with probability proportional to $\pi (\tau')$.
\end{enumerate}

{\bf Broken Circuit Complex.} For a concrete example, it turns out that the set NBC independent sets of any matroid $M$ (with respect to any ordering $\ord$) form a {\bf pure} simplicial complex that is known as the {\it broken circuit complex}. We denote this complex by $\NBC(M, \ord)$. We state purity as the following fact.
\begin{fact}\label{fact:purity}
    For every NBC independent set $I$, there exists an NBC base $B$ such that $I\subseteq B$.
\end{fact}
 The face numbers of the complex $BC(M,\ord)$ is the sequence $n_0, n_1,\dots,n_r$ where $n_i$ is the number of NBC independent sets of rank $i$. As alluded to above this sequence is in variant over $\ord$. The down-up walk over this complex equipped with a uniform distribution over its facets is the same as the down-up walk over NBC bases we explained before.
 
The link of a face $\tau\in X$ is the simplicial complex $X_{\tau} \coloneqq \{\sigma \setminus \tau : \sigma \in X, \sigma \supset \tau\}$. For each face $\tau$, we define the induced distribution $\pi_{\tau}$ on the facets of $X_{\tau}$ as 
\begin{align}\label{eq:pidef}
    \pi_{\tau}(\eta) = \Pr_{\sigma \sim \pi}[\sigma \supset \eta \mid \sigma \supset \tau]. 
\end{align}

{\bf Local Walks.} For any face $\tau$  of size $ 0 \leq k \leq d-2$, the local walk for $\tau$ is a Markov chain on the ground set of $X_{\tau}$ with transition probability matrix $P_{\tau}$ is defined as
\begin{align}\label{eq:localwalkdef}
    P_{\tau}(x,y)   = \frac{1}{d-k-1} \Pr_{\sigma \sim \pi_\tau}[y \in \sigma \mid \sigma \supset \tau \cup \{x\}]. %
\end{align}
for distinct $x,y$ in the ground set of $X_\tau$. 
The following theorem shows that the spectral expansion of the global walk $P^{\vee}$ on a simplicial complex can be bounded through bounding the local spectral expansion of the complex. 
\begin{theorem}[Local-to-Global Theorem \cite{DK17,KO18,DDFH18,AL20}]\label{thm:localtoglobal}
Say a $\dimension$-dimensional weighted simplicial complex $(X,\pi)$  is a $(\gamma_{0},\dots,\gamma_{\dimension-2})$-local spectral expander if for every face $\tau$ of size $0 \leq k\leq d-2$,  the second largest eigenvalue of $P_\tau$ is at most $\gamma_k$, i.e., $\lambda_{2}(P_{\tau}) \leq \gamma_{k}$.

Given a weighted simplical complex  $(X,\pi_{\dimension})$ 
that is a $(\gamma_{0},\dots,\gamma_{\dimension-2})$-local spectral expander, the down-up walk which samples from $\pi$ has spectral gap lower bounded by
\begin{align*}
    1- \lambda_2(P^{\vee}) \geq \frac{1}{\dimension} \prod_{j=0}^{\dimension-2} (1 - \gamma_{j})
\end{align*}
\end{theorem}

To prove that the down-up walk mixes rapidly on the bases of any matroid, \cite{ALOV19} proved that the independent set complex of any matroid $M$ is a $(0,0,\dots,0)$-local spectral expander. Building on this, a natural method to prove \cref{conj:mixingdownupNBC} is to show that the broken circuit complex of any matroid $M$ of rank $r$ and for any total ordering is a $(\gamma_{0},\dots,\gamma_{r-2})$-local spectral expander for $\gamma_i\leq \frac{O(1)}{r-i}$.
\begin{conjecture}\label{conj:BCClocalwalks}
    For any matroid $M$ of rank $r$ and any ordering $\ord$ the broken circuit complex of $M$ is a $(\gamma_{0},\dots,\gamma_{r-2})$-local spectral expander for some $\gamma_i\leq \frac{O(1)}{r-i}$
\end{conjecture}

\subsection{Our results}

Our main result is to disprove \cref{conj:BCClocalwalks} in a very strong form, namely for the class of (truncated) graphic matroids. 
\begin{restatable}{theorem}{THMnbccontractdownup}\label{thm:nbc-contract-down-up}
    There exists an infinite sequence of (truncated) graphic matroids $M_1,M_2,\dots$ with orderings $\ord_1, \ord_2, \dots$, such that for every $n\geq 1$, $M_n$ has  $\text{poly}(n)$ elements, and there exists a face $\tau$ of the broken circuit complex of $X=\NBC (M_n,\ord)$ for which the down-up walk on the facets of the link $X_\tau$ has a spectral gap of at most $n^{-\Omega(n)}$.
\end{restatable}

In fact, we even prove a stronger statement
\begin{restatable}{theorem}{THMnbccontractcount}\label{thm:nbc-contract-count}
    Given a matroid $M=(E,{\cal I})$ and a total ordering $\ord$ and a set $S\subseteq E$, unless RP=NP, there is no FPRAS for counting the number of  NBC bases of $M$ that contain $S$.
\end{restatable}
Although this theorem does not refute \cref{conj:mixingdownupNBC}, it shows that one probably need different techniques (or probably a different chain) to sample/count NBC bases of a matroid. Indeed, one may even need a different proof for the performance of down-up walk to sample ordinary bases of matroids. 

To complement our main results we also prove that, unlike optimization on bases of a matroid, optimization is NP-hard on the NBC bases of matroids. Moreover, unless NP= RP, there is no FPRAS for computing the sum of the weights of all NBC bases of a matroid subject to an external field, while the same computation over the bases of a matroid has a FPRAS.
\begin{restatable}{theorem}{THMoptnbc}\label{thm:opt-nbc}
Given a matroid $M=(E,{\cal I})$ with $|E|=n$ elements, an arbitrary total ordering $\ord$, and weights $w_1,\dots,w_n$, it is NP-hard to find the maximum weight NBC basis of $M$, where the weight of a NBC basis B is $\sum_{i\in B} w_i$.
\end{restatable}

\begin{restatable}{theorem}{THMnbcfieldcount}\label{thm:nbc-field-count}
Given a matroid $M=(E,{\cal I})$ with $|E|=n$ elements, a total ordering $\ord$, and weights $\{1\leq \lambda_e \leq O(n)\}_{e\in E}$, unless NP = RP, there is no FPRAS for computing the partition function of the $\lambda$-external field applied to uniform distribution of NBC independent sets, i.e., there is no FPRAS for computing :
$$ \sum_{B\text{ NBC Base}} \prod_{e \in B} \lambda_e.$$
\end{restatable}

It is well known that a $0/1$-polytope (i.e. the convex hull of a subset $S \subseteq \{0, 1\}^n$) has all vertices of equal hamming weight $r$ and edges of $\ell_2$ length $\sqrt{2}$ iff the polytope is a matroid base polytope of rank $r$ \cite{GELFAND1987301}. 
Moreover, assuming the Mihail-Vazirani conjecture, there is efficient algorithm to sample a uniformly random vertex of a $0/1$-polytope with constant sized edge length  \cite{MV89}.

We show that, unlike matroids, the NBC Base polytope, i.e. the convex hull of the indicator vectors of all NBC bases of a matroid $M$, has edges of arbitrarily long length. 
\begin{restatable}{theorem}{THMedgelength} \label{thm:edge-length}
    For any $n$, there exists a graphic matroid $M$ with $n$ elements and a total ordering $\ord$ such that the convex hull of all NBC bases of $M$ has edges of $\ell_2$ length at least $\Omega(\sqrt{n})$.
\end{restatable}

\section{Preliminaries}
Given a graph $G = (V, E)$, we denote the number of independent sets of size $i$ of $G$ by $i_k (G)$
For every set $S \subseteq V$, we define $N(S) \coloneqq  \{v\notin S: \exists u \in S, \{u,v\}\in E\}$ as the set of neighbors of $S$ in $G$.
\begin{definition}[Conductance]
    Given a weighted $d$-regular graph $G=(V,E,w)$, with weights $w:E\to\R_{\geq 0}$, for $S \subseteq V$, the conductance of $S$ is defined as
    \[
    \phi(S) = \frac{w(S,\overline{S})}{d|S|},
    \]
    where  $w(S,\overline{S})$ is the sum of the weights of edges in the cut $(S,\overline{S})$. Note that since $G$ is regular, the weighted degree of every vertex is $d$.
    The conductance of  $G$ is defined as
    \[
    \phi(G) = \min_{S:|S|\leq |V|/2} \phi(S).
    \]
\end{definition}
Given a weighted graph $G=(V,E,w)$, the simple random walk is the following stochastic process: Given $X_0=v\in V$, for every $u\sim v$, we have $X_1=u$ with probability $\frac{w_{\{u,v\}}}{d_w(v)}$  and we let $P$ be the transition probability matrix of the walk.

The following theorem is well-known and follows from the easy side of the Cheeger's inequality. 
\begin{theorem}\label{thm:bottlkeneck}
For any regular graph $G=(V,E)$ and any set $S\subseteq V$ and $|S|\leq |V|/2$
$$ \frac{1-\lambda_2(P)}{2}\leq \phi(G) \leq \phi(S) \leq \frac{|N(S)|}{|S|}$$
where $1-\lambda_2(P)$ is the spectral gap of the simple random walk on $G$.
\end{theorem}
A graphic matroid $M = (E, \mathcal{I})$ is a matroid defined on the edges of a graph $G = (V, E)$ and its independent sets are all subsets of edges that do not contain any cycle. It is easy to verify that circuits of $M$ correspond to cycles of $G$.  
\begin{definition}[Matroid Truncation]
    Let $M=(E, \mathcal{I})$ be a matroid of rank $r$. The truncation of $M$ to rank $r' \leq r$ removes all independent sets of size strictly greater than $r'$. It is easy to see that the truncation of any matroid $M$ to any $r'\leq r$ is also a matroid. 
\end{definition}

Let $M'$ be the truncation to rank $r'$ of a graphic matroid of rank $r$ defined on the edges of a graph $G$. The bases of $M'$ correspond to forests with $r'$ edges and the circuits of $M'$ are the circuits of $G$ along with all spanning forests of size $r'+1$.

The following fact about polytopes follows from convexity.
\begin{fact}\label{fact:polytopeedges}For any polytope $P\subseteq \R^d$ with vertices $v_1,\dots,v_n\in \R^d$, $\{v_i,v_j\}$ is an edge of $P$ iff there exists a weight function $w\in \R^d$ such that 
$$ \langle w,v_i\rangle=\langle w,v_j\rangle > \langle w,v_k\rangle,$$
for any $k\neq i,j$. 
\end{fact}

\section{Results}

We start with proving  \cref{thm:edge-length}.
\THMedgelength*
\begin{proof}%
Let $n$ be odd. Consider the following graphic matroid $M$ (with $n$ edges), with the ordering $\ord$: $1<2<\dots<n$  defined by the edges of the following graph:

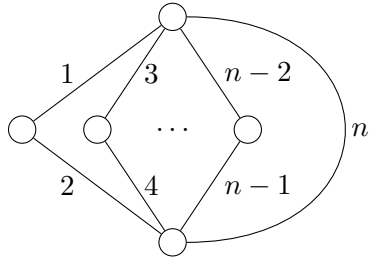
\begin{figure}[htb]\centering
\begin{tikzpicture}[every node/.style={minimum size = 2ex}]
\node[draw, circle] (t) at (0, 1.5) {};
\node[draw, circle] (b) at (0, -1.5) {};
\node[draw, circle] (m1) at (-2, 0) {};
\node[draw, circle] (m2) at (-1, 0) {};
\node[draw, circle] (m3) at (1, 0) {};

\draw (t) -- node[xshift = -2.2ex]{$1$} (m1) -- node[xshift = -2.2ex]{$2$} (b);
\draw (t) -- node[xshift = 1.2ex] {$3$} (m2) -- node[xshift = 1.2ex]{$4$} (b);
\draw (t) -- node[xshift = 0.3ex, anchor = west]{$n-2$} (m3) -- node[anchor = west, xshift = 0.3ex]{$n-1$} (b);
\draw[draw = none] (m2) -- node{$\ldots$} (m3);
\draw (t) .. node[xshift = 1.1ex]{$n$} controls (3,1.5) and (3, -1.5) .. (b);
\end{tikzpicture}
\label{fig:longedgegraph}
\end{figure}

We show that for $B = \{n\} \cup \{2i -1: 1\leq i\leq \frac{n-1}{2}\} $ and $B' = \{1\} \cup \{2i: 1\leq i\leq \frac{n-1}{2}\} $, $\{B, B'\}$ forms an edge in the NBC matroid base polytope denoted as $P_M$. We define a $w \in \mathbb{R}^n$ and then use \cref{fact:polytopeedges} to prove the statement.  Let $w_{n} = \frac{n+1}{2}$, and for any  $1\leq i \leq \frac{n-1}{2}$, let $w_{2i} = 1$ and $w_{2i-1} = 0$. It is easy to check that the function $\langle w, \bone_B \rangle = \langle w, \bone_{B'} \rangle = \frac{n+1}{2}$ and $\langle w, \bone_{B''} \rangle <\frac{n+1}{2} $ for all NBC basis $B'' \neq B, B'$. Therefore $\{B, B'\}$ forms and edge in $P_M$. The statements follows from the fact that $\|\mathbf{1}_B -\mathbf{1}_{B'}\|_2 =\sqrt{n}$.
\end{proof}

Next, we prove \cref{thm:opt-nbc} via a reduction from the  \textsc{MAX-INDEP-SET} problem: Given a graph $G= (V, E)$, a weight function $w: V \rightarrow \mathbb{R}_{\geq 0}$, and an integer $k$, decide whether $G$ has an independent set of weight at least $k$ or not. 

Note that independent sets of $G$ and independents sets of a BC complex/matroid are two different notions.
To complete the proof we use the following well-known hardness result.
\begin{theorem}[\cite{Karp1972}]
\textsc{MAX-INDEP-SET} is NP-complete. 
\end{theorem}

\THMoptnbc*
\begin{proof}%
    We prove this by a reduction from \textsc{MAX-INDEP-SET}. Let $G = (V, E)$ be a graph, a vertex weight function $w: V \to \R_{\geq 0}$ and $k$ an integer. Construct a new graph $G' = (V', E')$ from $G$ by first copying $G$ and then adding a new vertex $z$ and edges $e_v = \{z, v\}$ for all $v \in V$. We define $w': E' \to \R_{\geq 0 }$ as $w'(e_v) = w(v)$ for every $v \in V$, and $w'(e) = 0$ for every $e \in E$. 
    Moreover, consider the following total ordering $\ord$ on $E'$:
    $$E < \{e_v: v\in V\},$$ where the ordering within each set is arbitrary.
    Let $M$ be the graphic matroid defined by the edges of $G'$, we will be look at bases/independent sets of $\NBC( M, \ord)$.

    \begin{claim}
        There exists an independent set of $G$ of weight at least $k$ iff there exists an NBC basis of $M$ with weight at least $k$.
    \end{claim}
    
    We prove the claim in a straightforward manner. Suppose there is an independent set $I \subseteq V$ of $G$ with $w(I) \geq k$ and consider the set $I' \subseteq E'$ defined by $I' = \{e_v \, : \, v \in I\}$. By definition, $w'(I') \geq k$. We argue that $I'$ does not contain any broken circuit. Assume otherwise that there is a broken circuit $ C \setminus \{e\} \subseteq I'$. Since $C$ corresponds to a cycle in $G'$ and $C\setminus \{e\}$ is contained in $I'$, it is not hard to see that $C \setminus \{e\} = \{e_v, e_{v'}\}$ for some $v , v' \in I$ and $e = \{v, v'\}$ is an edge in $G$. But this is a contradiction with the fact that $I$ is an independent set of $G$. 
    Hence $I'$ is a NBC independent set. Since the broken circuit complex is pure (see \cref{fact:purity}), there exists an NBC basis $B$ containing $I'$ which has weight $w'(B) \geq w'(I') \geq k$. 
    
    For the other direction, suppose we have a NBC basis $B' \subseteq E'$ of weight  $w'(k) \geq k$, and define $I \subseteq V$ by $I = \{v \, : \, e_v \in B'\}$. Since all edges coming from $E$ have zero weight,  $w(I) = w'(B') \geq k$ . To see that $I$ is an independent set of $G'$, note that if there is an edge $\{v, v'\}$ for some $v, v' \in I$, we have $e_{v}, e_{v'} \in B'$, then $\{e_{v}, e_{v'}\}$ forms a broken circuit according to the ordering $\ord$. Therefore $I$ is an independent set of $G$ of weight at least $k$. 
\end{proof}

It's important to note that the above proof works under the crucial assumption that the order $\ord$ is chosen carefully based on the weights (and in some sense in the same order of the weights).  %

We can amplify the ideas in the previous construction to also argue \cref{thm:nbc-contract-down-up}. This is done by constructing a Broken Circuit complex for which the down-up walk of a carefully chosen link has inverse exponentially small spectral gap.

\begin{figure}[htb]\centering
\begin{tikzpicture}[every node/.style={minimum size = 2ex}]
\node[draw, circle] (z) at (0, 0) {$z$};
\node[draw, circle] (y) at (2, 0) {$y$};
\draw (y) -- node[anchor = south] {$e_0$} (z);

\node[draw, circle, inner sep = 0pt] (zv1) at (-3, -2) {$z_{v, 1}$};
\node[draw, circle, inner sep = 0pt] (zvl) at (-1, -2) {$z_{v, \ell}$};
\draw (z) -- node[shift = {(-1ex, -0.5ex)}, anchor = east]{$e_{v, 1}$} (zv1);
\draw (z) -- node[yshift = -0.5ex, anchor = east]{$e_{v, \ell}$} (zvl);
\draw[draw = none] (zv1) -- node{$\ldots$} (zvl);

\node[draw, circle, inner sep = 0pt] (zu1) at (1, -2) {$z_{u, 1}$};
\node[draw, circle, inner sep = 0pt] (zul) at (3, -2) {$z_{u, \ell}$};
\draw (z) -- node[yshift = -0.5ex, anchor = west]{$e_{u, 1}$} (zu1);
\draw (z) -- node[shift = {(0.5ex, -0.5ex)}, anchor = west]{$e_{u, \ell}$} (zul);
\draw[draw = none] (zu1) -- node{$\ldots$} (zul);

\draw[dotted] (0, -4) ellipse (4 and 1) node[yshift = -4ex] {$G$};
\node[draw, circle] (v) at (-2, -4) {$v$};
\draw (zv1) -- node[yshift = 1ex, anchor = east]{$f_{v, 1}$} (v);
\draw (zvl) -- node[shift = {(0.75ex, 1ex)}, anchor = east]{$f_{v, \ell}$} (v);

\node[draw, circle] (u) at (2, -4) {$u$};
\draw (zu1) -- node[shift = {(-0.75ex, 1ex)}, anchor = west]{$f_{u, 1}$} (u);
\draw (zul) -- node[yshift = 1ex, anchor = west]{$f_{u, \ell}$} (u);

\draw[draw = none] (u) -- node{$\ldots$} (v);
\end{tikzpicture}

\caption{A schematic of the graph $G'$ in the proofs of \cref{thm:nbc-contract-down-up} and \cref{thm:nbc-contract-count} where $G=K_{n,n}$ is the complete bipartite graph in the former and it is a hard instance of $\sharp\textsc{INDEP-SET-INC}(7, \frac{2}{19})$ in latter. }

\label{fig:nbcgenericconstruction}
\end{figure}
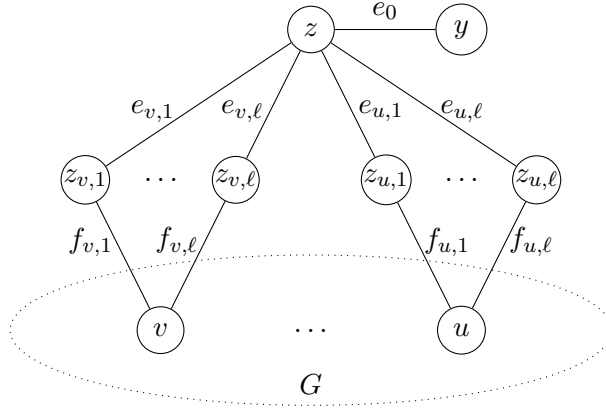

\THMnbccontractdownup*
\begin{proof} %
    Take the complete bipartite graph $G = K_{n, n}= (A,B, E=A\times B)$ , with  $|A|=|B|=n$. Also, let $V=A\cup B$.  Let $\ell\geq 1$ be a parameter that we choose later, and construct a new graph 
    $$G' = \left(V'=V\cup\{y,z\} \cup \{z_{ v, i}: v\in V,i\in [\ell]\}, E'=E\cup \{e_0\}\cup \left \{e_{v,i},f_{v,i} :v\in V,i\in [\ell]\right\}\right)$$
    where $e_0=\{y,z\}, e_{v,i}=\{z,z_{ v, i}\}, f_{v,i}=\{z_{v,i},v\}$ (see \cref{fig:nbcgenericconstruction}).
For a sanity check, note that $|V| = 2n$ and $|V'|=2\ell n + 2n+2$.
    
    Let $M = (E', \mathcal{I})$ be the graphic matroid defined by $G'$  truncated  to rank $2\ell n+ n +1$, i.e., the bases of $M$ are forests of $G'$ with exactly $2\ell n+n+1$ edges. Now, consider the following total ordering $\ord$ on $E'$:  
    $$ e_0 < E < \{e_{v,i}: v\in V,i\in [\ell]\} < \{f_{v,i}: v\in V, i\in [\ell]\},$$
    where the ordering within each set is arbitrary.
    Moreover, let $X \coloneqq \NBC( M, \ord)$, and define 
    $$\tau = \{e_{v, i}: v \in V, i \in [\ell]\}.$$ 
   For simplicity of notation, let $F_A \coloneqq \{f_{v, i}: v \in A, i\in [\ell]\}$ and  $F_B \coloneqq \{f_{v, i}: v \in B, i\in [\ell]\}$.
   \begin{claim}
       For any facet $S$ of $X_\tau$, either $S \cap F_A = \emptyset$, or $S \cap F_B = \emptyset$,
   \end{claim}
    This follows from the fact that $G$ is a complete bipartite graph and edges in $E$ are smaller than $e_{v,i}$'s and $f_{u,j}$'s; so if $S \cap F_A, S\cap F_B\neq\emptyset$, then it has a broken circuit.
    
    Therefore, the set of facets of $X_\tau$ can be partitioned into $2n+1$ sets $(\cup_{i=1}^n \mathcal{S}_{A, i}) \cup (\cup_{i=1}^n \mathcal{S}_{B, i}) \cup \mathcal{S}_{0}$, where  $\mathcal{S}_{A, i}$ is the set of all facets $S$ with $|S \cap F_A| = i$,  $\mathcal{S}_{B, i}$ is the set of all facets $S$ with $|S \cap F_B| = i$, and $\mathcal{S}_{0}$ is the set of all facets with $|S \cap (F_A \cup F_B)| = 0$. %
    Let $\mathcal{S}_A \coloneqq \cup_{i=1}^n \mathcal{S}_{A, i}$ and similarly define ${\cal S}_B$. We show that $\frac{|N(\mathcal{S}_A)|}{|\mathcal{S}_A|} \leq n^{-\Omega(n)}$, where $N({\cal S}_A)$ is the set of neighbors of ${\cal S}_A$ in the down-up walk $P^\vee_\tau$ on the facets of $\tau$. WLOG we can assume that $|\mathcal{S}_A|$ is at most half of all facets. Applying \cref{thm:bottlkeneck}, this would imply that $1 - \lambda_2(P^{\vee}_\tau)  \leq n^{-\Omega(n)}$. %
    
    First, note that for every  facet $S \in  \mathcal{S}_{A}$ and $T \in \mathcal{S}_{B}\smallsetminus {\cal S}_{B,1}$, we get $P^{\vee}(S, T) = 0$ since $|S\Delta T|>2$. So, $N(\mathcal{S}_A) \subseteq \mathcal{S}_{B, 1} \cup \mathcal{S}_{0}$. 
    First, notice $|\mathcal{S}_{0}|\leq {|E|\choose n}\leq  n^{2n}$. 
    Furthermore,
    $ |\mathcal{S}_{B, 1}| \leq {n \choose 1}\ell {|E| \choose n-1} \leq \ell n^{2n}$. 
    This follows from the fact that  any facet in $\mathcal{S}_{B, 1}$ can be written as $\{f_{v, i_v}\} \cup \{e_0\} \cup K$ for some $v \in A$,  $i_v \in [\ell]$, and subset $K \subseteq E$ of size $n-1$. 
    Lastly, $|\mathcal{S}_A| \geq |\mathcal{S}_{A, n}| = \ell^{n}$. 
    This is because every choice of $\{i_v\}_{v \in A}$ corresponds to a set in $\mathcal{S}_{A, n}$ whose sets are of the form $\{f_{v, i_v}: v \in V\} \cup \{e_0\}$. These sets all don't contain a broken circuit because the circuits introduced through truncation are exactly the forests with $2\ell n + n + 2$ edges. However, any proper superset of $\{f_{v, i_v}: v \in V\} \cup \{e_0\}$ must include $e_0$, so looking at the circuit introduced by the superset, the corresponding broken circuit will always remove $e_0$.
    Putting it all together, $$1-\lambda_2(P^\vee_\tau)\leq \frac{|N(\mathcal{S}_A)|}{|\mathcal{S}_A|} \leq \frac{n^{2n}(1+\ell)}{\ell^n} \underset{\text{assuming }\ell\geq n^3}{\leq} n^{-\Omega(n)}.$$ as desired.

\end{proof}
We prove \cref{thm:nbc-field-count} and \cref{thm:nbc-contract-count} by a reduction from  $\sharp\textsc{INDEP-SET-INC}(7, \frac{2}{19})$, defined as the following. 
\begin{definition}[$\sharp\textsc{INDEP-SET-INC}(7, \frac{2}{19})$]\label{def:INDEPSETINC}
Given a $7$-regular graph $G = (V, E)$ that satisfies $i_{k} (G) \leq i_{\lfloor \frac{2|V|}{19}\rfloor} (G)  $  for any $k<\lfloor \frac{2|V|}{19}\rfloor $, where $i_k(G)$ are the independent sets of $G$ of size $k$, count the number of independent sets of size $\lfloor \frac{2|V|}{19}\rfloor $. 
\end{definition}
\begin{theorem}\label{thm:NPindpset}
        Unless $\mathrm{NP}=\mathrm{RP}$, there is no randomized algorithm with constant approximation ratio for $\sharp\textsc{INDEP-SET-INC}(7, \frac{2}{19})$.
\end{theorem}
We leave the proof of this for the appendix.
Now, we are ready to prove \cref{thm:nbc-contract-count}. The high-level structure of the proof is similar to the proof of \cref{thm:nbc-contract-down-up} where we apply a similar gadget to graphs on which it is hard to count independent sets (as opposed to the complete bipartite graph).

\THMnbccontractcount*
\begin{proof}%
For simplicity of notion, let $\alpha \coloneqq \frac{2}{19}$. 
    We prove by a reduction from  $\sharp\textsc{INDEP-SET-INC}(7, \frac{2}{19})$. Take any arbitrary $7$-regular graph $G = (V, E)$ whose number of independent sets of size $\lfloor \alpha|V|\rfloor$ is at least  the number of its independent sets of size $k$  for any $k<\lfloor \alpha|V|\rfloor$. Let $n \coloneqq |V|$ and $N$ be the number of independent sets of size $\lfloor\alpha n\rfloor$ of $G$. Also, define  $\ell \geq 1$ to be a parameter that we choose later. %
    
    Now, construct a new graph
    $$G' = \left(V'=V\cup\{y,z\} \cup \{z_{ v, i}: v\in V,i\in [\ell]\}, E'=E\cup \{e_0\}\cup \left \{e_{v,i},f_{v,i}:v\in V,i\in [\ell]\right\}\right)$$
    where $e_0=\{y,z\}, e_{v,i}=\{z,z_{ v, i}\}, f_{v,i}=\{z_{v,i},v\}$ (see \cref{fig:nbcgenericconstruction}).    
   Let $M = (E', \mathcal{I})$ be the graphic matroid defined by $G$  truncated  at rank $\ell n+\lfloor\alpha n\rfloor+1$, i.e., the bases of $M$ are forests of $G'$ with exactly $\ell n+\lfloor\alpha n\rfloor+1$ edges. Now, consider the following ordering $\ord$ on $E'$:  
    $$ e_0 < E < \{e_{v,i}: v\in V,i\in [\ell]\} < \{f_{v,i}: v\in V, i\in [\ell]\},$$
    where the ordering within each set is arbitrary.  Moreover, let $X \coloneqq \NBC( M, \ord)$, and define 
    $$\tau = \{e_{v, i}: v \in V, i \in [\ell]\}.$$ 
  We claim that the number of facets of $X_\tau$ is at least $\ell^{\lfloor\alpha n\rfloor} N$ and at most  $2\ell^{\lfloor\alpha n \rfloor}  N$. 
  So, a $1.5$-approximation to the number facets of $X_\tau$, i.e., the number NBC bases of $M$ that contain $\tau$, gives a $3$-approximation to $N$, the number of independent sets of size $\lfloor \alpha n\rfloor$ of $G$.
  
  We use the following crucial observation:
  \begin{claim} For any facet $S$ of $X_\tau$, $\{v: \exists f_{v,i}\in S\}$ is an independent set of $G$ and for any $f_{v, i}, f_{v, j} \in S$ we have $i=j$.
  
  Conversely, for any $S\subseteq \{ f_{v, i}: v \in V, i \in [\ell]\}$, such that the set $\{v: \exists f_{v, i}\in S\}$ is an independent set of size $\lfloor\alpha n\rfloor$ of $G$, and  $f_{v, i}, f_{v, j} \in S\implies i=j$, we have $S\cup \{e_0\}$ is a facet of $X_\tau$. 
  \end{claim}
  The proof simply follows from the fact that edges of $E$ are smaller than $e_{v,i}$'s, and $f_{u,j}'s$ in $\ord$.
  By the second part of the claim, we can write 
    \begin{align}\label{eq:reduction1}
           |X_\tau (\lfloor\alpha n\rfloor+1)| &= \ell^{ \lfloor\alpha n\rfloor}  N +  |\{ S \in X_\tau (\lfloor\alpha n\rfloor+1): S \cap E \neq \emptyset \} | \geq \ell^{\lfloor \alpha n\rfloor} N.
    \end{align}
 Define $i_k \coloneqq i_k (G)$ as the number of independent sets of size $k$ of graph $G$. By the first part of the above claim we can write,
    \begin{align}\label{eq:bound-bad-bases}
  |\{ S \in X_\tau (\lfloor\alpha n\rfloor+1): S \cap E \neq \emptyset \} | &\leq \sum_{k=0}^{\lfloor\alpha n\rfloor-1} \ell^k \cdot i_k\cdot {|E| \choose \lfloor\alpha n\rfloor-k} \leq \sum_{k=0}^{\lfloor\alpha n\rfloor-1} \ell^k \cdot i_k\cdot |E| %
  ^{\lfloor\alpha n\rfloor-k}   \\
    &\underset{\text{using }i_k\leq N}\leq N|E|^{\lfloor\alpha n\rfloor}\sum_{k=0}^{\lfloor \alpha n\rfloor-1} (\ell/|E|)^{k}\\
  &\underset{\text{assuming }\ell\geq 2|E|}{\leq} N |E|^{\lfloor\alpha n\rfloor} (\ell/|E|)^{\lfloor \alpha n\rfloor}\leq N\ell^{\lfloor \alpha n\rfloor}
       \end{align}
Putting these together with \eqref{eq:reduction1} concludes the proof.
\end{proof}

\THMnbcfieldcount*
\begin{proof}%
For simplicity of notion, let $\alpha \coloneqq \frac{2}{19}$. 
    The proof is similar to the proof of \cref{thm:nbc-contract-count} by a reduction from  $\sharp\textsc{INDEP-SET-INC}(7, \frac{2}{19})$. Take any arbitrary $7$-regular graph $G = (V, E)$ with $n:=|V|$ vertices whose number of independent sets of size $\lfloor \alpha|V|\rfloor$ is at least  the number of its independent sets of size $k$  for any $k<\lfloor \alpha|V|\rfloor$.
    Construct a new graph 
    $$G' = (V'=V\cup \{y,z\}, E'=E\cup \{e_0=\{y,z\}\}\cup\{e_v=\{v,z\}: v\in V\})$$ 
    Let $M = (E', \mathcal{I})$ be the graphic matroid given by $G'$ truncated to rank $ \lfloor\alpha n\rfloor +1$ and consider the following ordering $\ord$ on $E''$:  
    $$ e_0 < E < \{e_{v}: v\in V\},$$
    where as usual the ordering within each set is arbitrary. Define weights $\lambda: E'\to\R_{\geq 0}$ as follows:
    $$\lambda_{e} = \begin{cases} \ell& \text{if }e=e_v \text{ for some $v\in V$},\\
    1&\text{o.w.}\end{cases},$$
    for some $\ell$ that we choose later.
    We argue that 
    $$\lambda^{\lfloor\alpha n\rfloor}  N \leq  \sum_{B} \prod_{e \in B} \lambda_e \leq 2\lambda^{\lfloor\alpha n \rfloor}  N.$$ 
    where here (and henceforth) the sum is over $B$'s that are NBC bases of $M$, and therefore a $1.5$-approximation to the partition function, i.e., the quantity in the middle, is a 3-approximation to $N$. Similar to the previous theorem we have the following claim.
    \begin{claim}
        For any NBC base $B$ of $M$, we have $\{v: e_v\in B\}$ is an independent set of $G$. 
        Conversely, for any independent set $I$ of $G$ of size $|I|=\lfloor\alpha n\rfloor$, $\{e_0\}\cup\{e_v: v\in I\}$ is a NBC base of $M$.
    \end{claim}
So,
    \begin{align} \label{eq:reduction2}
           \sum_B \prod_{e \in B} \lambda_e  &= \sum_{B: B\cap E \neq \emptyset} \prod_{e \in B} \lambda_e +  \sum_{B: B\cap E = \emptyset} \prod_{e \in B} \lambda_e 
          \\
           &= \sum_{B: B\cap E \neq \emptyset} \prod_{e \in B} \lambda_e + \ell^{\lfloor\alpha n\rfloor} |\{S \subseteq V: S \text{ independent set of } G, |S| = \lfloor\alpha n\rfloor\}|\nonumber
    \end{align}
Define $i_k$ as the number of independent sets of size $k$ of graph $G$. We have
    \begin{align*}
\sum_{B: B\cap E \neq \emptyset} \prod_{e \in B} \lambda_e \leq \sum_{k=0}^{\lfloor\alpha n\rfloor-1} \ell^k i_k {|E| \choose \lfloor\alpha n\rfloor-k} \underset{\substack{\text{using }i_k\leq N,\\ \text{assuming }\ell\geq 2|E|}}{\leq} \ell^{\lfloor\alpha n\rfloor}N %
       \end{align*}
where the last inequality follows from the same calculations as in  \cref{eq:bound-bad-bases}. %
\end{proof}
\printbibliography
\appendix

\section{Proof of \cref{thm:NPindpset}}
     In this section we prove \cref{thm:NPindpset}. We use a reduction from the problem of computing the partition function of the Hardcore model when the fugacity is above the critical threshold. Define $\sharp\textsc{HC}( \Delta, \lambda)$ as follows: given a $\Delta$-regular graph $G = (V, E)$, compute the partition function $Z_G(\lambda) = \sum_{I} \lambda^{|I|}$, where the sum is taken over the family of independent sets $I \subseteq V$ of $G$. The critical threshold is defined as $\lambda_c (\Delta) \coloneqq \frac{(\Delta-1)^{\Delta-1}}{(\Delta-2)^\Delta}$. 
\begin{theorem} [\cite{Sly10, Sly14, GSV16}]
      The following holds for any fixed $\epsilon > 0$, integer $\Delta \geq 3$  and $\lambda > \lambda_c(\Delta)$: unless NP=RP, for any $\lambda>\lambda_c(\Delta)$ there is no polynomial-time algorithm for for approximating $\sharp\textsc{HC}( \Delta, \lambda)$ up to a $1+\epsilon$ multiplicative factor. 
\end{theorem}
     
We give a polynomial-time algorithm that given a $e^{\pm \epsilon/2}$-approximation for  $\sharp\textsc{INDEP-SET-INC}(7, \frac{2}{19})$ (see \cref{def:INDEPSETINC}), approximates $\sharp\textsc{HC}( 7, \frac{2}{3})$ up to a $e^{\pm\epsilon}$-multiplicative error. Since $\frac{2}{3} > \lambda_c(7) = \frac{6^6}{5^7} \geq 0.6$, this  finishes the proof of \cref{thm:NPindpset}. Our reduction is a modification of Theorem 16 in \cite{DP21}. 
     \begin{theorem}
         There exists a polynomial-time algorithm that for any given $\epsilon \leq 1$, satisfies the following properties:
\begin{enumerate}
\item Given an instance $G = (V, E)$ of $\sharp\textsc{HC}( 7, \frac{2}{3})$, the algorithm constructs an instance $G' = (V', E')$ of the problem $\sharp\textsc{INDEP-SET-INC}( 7, \frac{2}{19})$
with size polynomial in $|G|$.
\item Given a  $e^{\pm\epsilon/2}$-multiplicative approximation to the number of independent sets of size $\lfloor \frac{2|V'|}{19} \rfloor$ of $G'$, a   $e^{\pm\epsilon}$-approximation of $Z_G(\frac{2}{3})$
can be computed in polynomial time.
\end{enumerate}
\end{theorem}
\begin{proof}
    
    Given a $7$-regular graph $G = (V, E)$, we define $G'$ as the disjoint union of $G$ with $r := \frac{c^2 n^2}{\epsilon}  $ copies of the complete graph $K_8$, where $n=|V|$, for some $c>1$ that we choose later. %
    For simplicity of notation, let %
    $N\coloneqq |V'|=n+8r$, $\alpha \coloneqq \frac{2}{19}$, $\lambda := \frac{2}{3}$. It is enough to show that $G'$ is an instance of $\sharp\textsc{INDEP-SET-INC}( 7, \frac{2}{19})$ and 
    \begin{align}\label{eq:hardness-hc}
        e^{-\epsilon/2} \frac{i_{\lfloor \alpha N\rfloor}\left(G^{\prime}\right)}{ { r\choose \lfloor \alpha N\rfloor} 8^{\lfloor \alpha N\rfloor}} \leq 
    Z_G(\lambda) \leq  e^{\epsilon/2} 
    \frac{i_{\lfloor \alpha N \rfloor}\left(G^{\prime}\right)}{ { r\choose \lfloor \alpha N\rfloor} 8^{\lfloor \alpha N \rfloor}},
    \end{align}
     where as usual $i_k(G)$ is the number of independent sets of size $k$ in $G$, and 
     $${ r\choose \lfloor \alpha N \rfloor} 8^{\lfloor \alpha N\rfloor} = i_{\lfloor \alpha N \rfloor}(r K_8). $$ 
      Here, $rK_8$ is a shorthand for the graph which is a disjoint union of $r$ copies of $K_8$.
    We first show that  \cref{eq:hardness-hc} holds. Note that $$i_{\lfloor \alpha N\rfloor}\left(G^{\prime}\right)=\sum_{j=0}^n i_j(G) i_{\lfloor \alpha N\rfloor-j}(r K_8)=i_{\lfloor \alpha N\rfloor}(r K_8) \sum_{j=0}^n i_j(G) \frac{i_{\lfloor \alpha N\rfloor-j}(r K_8)}{i_{\lfloor \alpha N\rfloor}(r K_8)}.$$
    Thus, to show  \cref{eq:hardness-hc}, it is enough to prove that for every $1 \leq j \leq n$, 
    \begin{align}\label{eq:hardness-hc2}
        e^{-\epsilon/2} \cdot \frac{i_{\lfloor \alpha N\rfloor-j}(r K_8)}{i_{\lfloor \alpha N\rfloor}(r K_8)} \leq \lambda^j \leq  e^{\epsilon/2}  \cdot \frac{i_{\lfloor \alpha N\rfloor-j}(r K_8)}{i_{\lfloor \alpha N\rfloor}(r K_8)} .
    \end{align}
 We can write 
\begin{align}\label{eq:hardness-hc3}
      \frac{i_{\lfloor \alpha N\rfloor-j}(r K_8)}{i_{\lfloor \alpha N\rfloor}(r K_8)}  = \frac{{ r\choose \lfloor \alpha N\rfloor -j} 8^{\lfloor \alpha N\rfloor-j}}{{ r\choose \lfloor \alpha N\rfloor} 8^{\lfloor \alpha N\rfloor} } = \frac{1}{8^j}\prod_{i=0}^{j-1}\frac{\lfloor \alpha N\rfloor -i}{r - \lfloor \alpha N\rfloor+j-i}.
\end{align}
To prove the upper bound, first note that
\begin{align}\label{eq:hardness-hc4}
    \frac{\alpha N }{r - \alpha N +j}  \underset{\alpha N \geq 8\alpha r}{\geq} \frac{ 8\alpha r}{r (1- 8\alpha) +n}\underset{\substack{n = \sqrt{\epsilon r}/c \\  \alpha = 2/19 }}{=}  \tfrac{16}{3}\Bigg(\frac{1}{1+\frac{19\sqrt{\epsilon}}{3c\sqrt{r}}}\Bigg).
\end{align}
This implies that $  \frac{\alpha N }{r - \alpha N +j} \geq 1$. So, $   \frac{\alpha N }{r - \alpha N +j}   \leq  \frac{\lfloor\alpha N \rfloor -i}{r - \lfloor\alpha N \rfloor +j - i} $ for every $i < r - \lfloor\alpha N \rfloor +j$. Thus, 

\begin{align*}\frac{1}{8^j}\cdot\prod_{i=0}^{j-1}\frac{\lfloor \alpha N\rfloor -i}{r - \lfloor \alpha N\rfloor+j-i} &\geq \frac{1}{8^j}\cdot\left(   \frac{\alpha N }{r - \alpha N +j} \right)^j  \\
&\underset{\substack{\cref{eq:hardness-hc4} \\ j \leq n = \sqrt{\epsilon r}/c }}{\geq}  \frac{1}{8^j}\cdot(\tfrac{16}{3})^j\Bigg(\frac{1}{1+\frac{19\sqrt{\epsilon}}{3c\sqrt{r}}}\Bigg)^{\sqrt{\epsilon r}/c} \geq (\tfrac{2}{3})^j e^{-\epsilon/2} = \lambda^je^{-\epsilon/2},
\end{align*}
for a large enough $c>1$.
Combining this with \cref{eq:hardness-hc3}, we get the upper bound in \cref{eq:hardness-hc2}.

    To prove the lower bound, note that 
    \begin{align*}
        \frac{1}{8^j}\cdot\prod_{i=0}^{j-1} \frac{\lfloor \alpha N\rfloor -i}{r - \lfloor \alpha N\rfloor+j-i} &\underset{j-i\geq 0}{\leq} \frac{1}{8^j}\cdot\left(\frac{\lfloor \alpha N\rfloor }{r - \lfloor \alpha N\rfloor}\right)^j \underset{\substack{ \lfloor \alpha N\rfloor 
        = \lfloor \frac{16r}{19} + \frac{2\sqrt{\epsilon r}}{19c}\rfloor}}{\leq} \frac{1}{8^j}\cdot \left(\frac{ \frac{16r}{19}(1 + \frac{\sqrt{\epsilon}}{8c\sqrt{r}})}{\frac{3r}{19}(1 - \frac{2\sqrt{\epsilon}}{3c\sqrt{r}})}\right)^j 
        \\
        &\leq (\tfrac{2}{3})^j e^{\epsilon/2}= \lambda^j\cdot e^{\epsilon/2},%
    \end{align*}
    for a large enough $c>1$.
    Combining this with \cref{eq:hardness-hc3}, the lower bound in \cref{eq:hardness-hc2}, thus \eqref{eq:hardness-hc} follows. %
    
    It remains to show that $G'$ is an instance of $\sharp\textsc{INDEP-SET-INC}( 7, \frac{2}{19})$, i.e. $i_{k} (G') \leq i_{\lfloor \alpha N\rfloor} (G')  $  for any $k<\lfloor \alpha N\rfloor $.
     For any $k<\lfloor \alpha N\rfloor$, and any independent set $S$  in the original graph $G$, let $T_{S, k}$ be the set of all independent sets of size $k$ of $G'$  whose intersection with the vertices of $G$ is $S$.  It is enough to show that there exists a constant $n_0$ such that if $n\geq n_0$,  then we have $|T_{S, k}| \leq |T'_{S, \lfloor \alpha N\rfloor}|$ for every independent set $S \subseteq V$ of $G$ and $k<\lfloor \alpha N\rfloor$. We prove a stronger statement that there exists a constant $n_0$ such that if $n \geq n_0$, then for any fixed independent set $S \subseteq V$,  $|T_{S, k}|$ is increasing as a function of $k$ for all $k \leq  \lfloor \alpha N\rfloor $. 
     It is enough to show that  $\frac{|T_{S, k}|}{|T_{S, k-1}|} \geq 1$ for any $|S| \leq k \leq  \alpha N$. Note that $|T_{S, k}| = { r\choose k-|S|} 8^{k-|S|}$. So we have
     $$\frac{|T_{S, k}|}{|T_{S, k-1}|}  = \frac{{r\choose  k-|S|} 8^{k-|S|}}{{r\choose  k-1-|S|} 8^{k-1-|S|}} = 8 \cdot \frac{r-k+|S|+1}{k-|S|} \geq 8 \cdot \frac{r-k}{k} \geq 8 \frac{\frac{3r}{19} -n}{\frac{16r}{19} + n}, $$
    where the last inequality comes from the fact that $k \leq \frac{2N}{19} = \frac{2}{19} (8r +n)  \leq   \frac{16r}{19} + n$. But since  $r = \frac{c^2n^2}{\epsilon}$, there  is a constant $n_0$ such that  for $n \geq n_0$, we have  $\frac{ \frac{3r}{19} -n}{\frac{16r}{19} + n} \geq \frac{1}{8}$. This shows that $\frac{|T_{S, k}|}{|T_{S, k-1}|} \geq 1$, which  finishes the proof. 
\end{proof}

\end{document}